\documentclass{article}

\usepackage{amsmath,amsthm,amsfonts,braket,graphicx}

\theoremstyle{definition} 
\theoremstyle{definition}

\newtheorem {lemma} {Lemma}

\title{Security Proof of a Semi-Quantum Key Distribution Protocol: Extended Version\footnote{This is an extended version of a paper published in IEEE ISIT 2015.}}
\author{Walter O. Krawec\\\small{Stevens Institute of Technology}\\\small{Hoboken NJ, 07030 USA}\\\small{\texttt{walter.krawec@gmail.com}}}

\begin{document}
\maketitle

\begin{abstract}
Semi-quantum key distribution protocols are designed to allow two users to establish a secure secret key when one of the two users is limited to performing certain ``classical'' operations.  There have been several such protocols developed recently, however, due to their reliance on a two-way quantum communication channel (and thus, the attacker's opportunity to interact with the qubit twice), their security analysis is difficult and little is known concerning how secure they are compared to their fully quantum counterparts.  In this paper we prove the unconditional security of a particular semi-quantum protocol.  We derive an expression for the key rate of this protocol, in the asymptotic scenario, as a function of the quantum channel's noise.  Finally, we will show that this semi-quantum protocol can tolerate a maximal noise level comparable to certain fully quantum protocols.
\end{abstract}

\section{Introduction}

Quantum Key Distribution (QKD) protocols are designed to allow two users, Alice ``$A$'' and Bob ``$B$'', to establish a shared secret key, secure against even an all powerful adversary, Eve ``$E$''.  Generally these protocols consist of three stages: first, the quantum communication stage, consisting of several iterations of $A$ and $B$ communicating over a quantum channel; for instance $A$ preparing qubits in a variety of different bases, and $B$ measuring them in different bases - if the basis choice matches, they share a secret bit of information.  Next, using an authenticated classical channel, one which $A$ and $B$ may read/write to, but $E$ may only read from, they perform a parameter estimation stage where they divulge a randomly chosen subset of their measurement results (along with other data depending on the protocol) in order to estimate various statistics of the quantum channel (in particular, the error rate of the channel).  The remaining measurement results, those which were not publicly divulged, are used as their \emph{raw key} - a string of bits which $E$ may hold some information on, and which may contain some discrepancies due to the error rate in the quantum channel.  If this error rate is ``low enough'' $A$ and $B$ will move on to the third stage, running an error correction protocol and privacy amplification protocol.  These two routines, which again utilize the authenticated classical channel, will first fix any errors in $A$ and $B$'s raw key, while privacy amplification will then take as input this error corrected raw key and output a key that is smaller, but secure.  The reader is referred to \cite{QKD-survey} for more information.

A major question is how low is ``low enough'' - that is how much noise can a protocol tolerate before $E$ potentially holds too much information on $A$ and $B$'s raw key so that privacy amplification is unable to distill a secure secret key.  To determine this, one typically considers the \emph{key rate} of a QKD protocol, which we denote by $r$.  This quantity is the fraction of secure secret key bits after privacy amplification, divided by the size of the raw key.

Let us first consider \emph{collective attacks} \cite{QKD-survey}: those attacks where an eavesdropper performs the same attack operation each iteration of the quantum communication stage (thus each iteration may be treated independently of the others), however $E$ is allowed to postpone her measurement until any future time of her choice.  Let $N$ be the size of $A$ and $B$'s raw key (before error correction, but after parameter estimation).  Let $\ell(N) \le N$ denote the number of secure secret key bits that $A$ and $B$ may distill after error correction and privacy amplification (possibly $\ell(N) = 0$ if no secure key bits can be distilled).  It was shown in \cite{QKD-renner-keyrate} and \cite{QKD-Winter-keyrate}, that, in the \emph{asymptotic scenario} (where $N \rightarrow \infty$):
\[
r = \lim_{N\rightarrow \infty}\frac{\ell(N)}{N} \ge \inf(S(B|E) - H(B|A)),
\]
where $S(B|E)$ is the conditional von Neumann entropy, $H(B|A)$ is the conditional Shannon entropy (see next section), and the infimum is over all collective attacks an eavesdropper may perform which conform to the statistics observed during parameter estimation (e.g., those attacks which induce the observed error rate).  Note that above, we have written the equation when \emph{reverse reconciliation} is used which seems more natural for the protocol we consider in this paper.  If the protocol in question is permutation invariant, it was shown in \cite{QKD-general-attack,QKD-general-attack2} that proving security against these collective attacks is sufficient to prove security against \emph{general attacks} - those attacks where $E$ is allowed to perform any operation allowed within the laws of physics.  The reader is again referred to \cite{QKD-survey} for more information on this key rate equation, direct and reverse reconciliation, the asymptotic scenario, and the various attack models commonly used.

Clearly it is desired that $r > 0$, in which case $A$ and $B$ may distill a secure secret key; if $r = 0$, $E$'s attack was too strong and she holds too much information on $A$ and $B$'s raw key for privacy amplification to work.  The goal then, typically, is to derive a simplified expression for the quantity $\inf(S(A|E) - H(A|B))$, given an observed error rate.  This expression should be a function only of values that may be estimated by $A$ and $B$ during the parameter estimation stage.

\subsection{Semi-Quantum Key Distribution}

A semi-quantum key distribution (SQKD) protocol, first introduced in \cite{SQKD-first}, has the same goal as a QKD protocol; however now, one of the two users - typically $B$ - is limited to performing what are called ``classical'' or ``semi-quantum'' operations.  Namely, $B$ is only allowed to work directly with the computational $Z$ basis (spanned by elements $\ket{0}$ and $\ket{1}$).  These protocols are interesting from a theoretical stand-point as they attempt to answer how ``quantum'' a protocol needs to be in order to obtain a benefit over a classical protocol.  Further, they may also be practically interesting as they may require less hardware on the limited user's side ($B$).

These protocols rely on a two-way quantum communication channel, allowing a qubit to travel from $A$ to $B$, then back to $A$.  The all-powerful attacker $E$, who is sitting between $A$ and $B$, is able now to attack the qubit twice making the security analysis of these protocols very difficult.  Up to now, most work has been showing the \emph{robustness} of a SQKD protocol.  As defined in \cite{SQKD-first,SQKD-second}, an SQKD protocol is robust if, for any attack that $E$ may perform which potentially causes her to gain information on $A$ or $B$'s raw key, necessarily induces a detectable disturbance.  Though this provides a good starting point for proving security, it is not as powerful a result as computing the key rate equation.  Robustness only tells the users that any noise might be an attacker; the latter notion gives a relation between the amount of noise induced by the most powerful class of attacks $E$ may employ, and the size of $A$ and $B$'s final secret key (possibly $\ell(N) = 0$ if there is too much noise).

While several authors have developed SQKD protocols and proven their robustness (for instance \cite{SQKD-first,SQKD-second,SQKD-3,SQKD-lessthan4,SQKD-cl-A}), only a few authors have considered security of an SQKD protocol beyond robustness.  In particular, \cite{SQKD-information} considered security against \emph{individual attacks} (a class of attacks weaker than collective attacks and which do not provide a security bound against general attacks - our primary goal); that reference managed to bound the amount of information $E$ may hold based on the level of noise induced by her (individual) attack.  In \cite{SQKD-cl-A}, the authors also considered a form of individual attack when proving security of a new protocol they devised.  In \cite{SQKD-MultiUser}, a \emph{mediated SQKD protocol} was developed and its security against general attacks was proven (including its key-rate equation); this is different from the setting in this paper, however, as such a mediated protocol operates with two semi-quantum users, and an untrusted quantum server.  Finally, in \cite{SQKD-Single-Security}, a series of security results were shown for a particular class of SQKD protocol: \emph{single state} protocols - ones where $A$ is limited to preparing and sending $\ket{+}$ each iteration of the quantum communication stage.  In this paper, however, we consider the more complicated \emph{multi-state} protocols where $A$ is not limited in this manner (she may prepare any qubit she likes each iteration, chosen randomly).

In this paper, we will prove, for the first time, the unconditional security of the SQKD protocol from \cite{SQKD-first}.  In particular, we will derive a lower bound on the key rate, based only on certain parameters $A$ and $B$ may estimate.  We will show that this protocol can tolerate up to $5.34\%$ noise before $A$ and $B$ should abort (that is, the key rate is strictly positive so long as the error rate is less than $5.34\%$).  While we cannot compare this noise threshold with other SQKD protocols (as none others have yet been considered, besides the mediated one mentioned earlier), this does compare favorably with certain ``fully'' quantum protocols - e.g., B92 \cite{QKD-B92} which supports up to $4.8\%$ noise \cite{QKD-renner-keyrate} and a three-state variant of BB84 which supports up to $4.25\%$ noise \cite{QKD-BB84-three-state}.  Thus our results provide further evidence that security in the semi-quantum setting can be comparable to security in the ``fully'' quantum setting.

\section{Notation}

Given a set of real values $\{p_1, p_2, \cdots, p_n\}$, with $p_i \ge 0$ and $\sum_ip_i = 1$, we write $H(p_1, p_2, \cdots, p_n)$ to be the classical or Shannon entropy of these values: $H(p_1, \cdots, p_n) = -\sum_ip_i\log_2 p_i$ (where, as usual, $0\cdot \log_2 0 = 0$).  When given only a single value $p \in [0,1]$, we write $h(p)$ to mean $H(p, 1-p)$.

Let $\rho$ be a density operator acting on some finite dimensional Hilbert space $\mathcal{H}$.  Then, we denote by $S(\rho)$ its von Neumann entropy.  Given the eigenvalues $\{\lambda_i\}$ of $\rho$, this value is: $S(\rho) = -\sum_i\lambda_i \log \lambda_i$.  (Note that all logarithms in this paper are base two, unless otherwise specified.)

If $\rho_{AB}$ is a density operator acting on the bipartite space $\mathcal{H}_A\otimes\mathcal{H}_B$, we will often write $S(AB)$ to denote the von Neumann entropy of $\rho_{AB}$ and $S(B)$ the von Neumann entropy of $\rho_B$ where $\rho_B = tr_A\rho_{AB}$.  We denote by $S(A|B)$ the von Neumann entropy of $A$'s system conditioned by $B$.  That is: $S(A|B) = S(AB) - S(B) = S(\rho_{AB}) - S(tr_A\rho_{AB})$.

Given a number $z \in \mathbb{C}$, we denote by $Re(z)$ and $Im(z)$, its real and imaginary components respectively.  The conjugate of $z$ is denoted $z^*$.  If $U$ is a matrix with complex entries, $U^*$ denotes its conjugate transpose.

We define the computational $Z$ basis to be those states $\{\ket{0}, \ket{1}\}$.  We define the Hadamard $X$ basis to be those states $\{\ket{+}, \ket{-}\}$ where:
\[
\ket{\pm} = \frac{1}{\sqrt{2}}\ket{0} \pm \frac{1}{\sqrt{2}}\ket{1}.
\]

\section{The Protocol}

The protocol we consider is a semi-quantum one, utilizing a two-way quantum channel (allowing a qubit to travel from $A$ to $B$, then back to $A$), with $B$ being the limited semi-quantum or ``classical'' user.  This means that, while $A$ may prepare and measure qubits in any basis of her choice (choosing different bases each iteration of the quantum communication stage), $B$ is limited to performing one of two operations each iteration:

\begin{enumerate}
  \item He may \emph{Measure and Resend}: the incoming qubit is subjected to a $Z = \{\ket{0}, \ket{1}\}$ basis measurement.  $B$ will then resend his measurement result to $A$ (i.e., if he measures $\ket{r}$, for $r \in \{0,1\}$, he will send a new qubit of the form $\ket{r}$ back to $A$).
  \item Or he may \emph{Reflect}: the incoming qubit is ignored by $B$ and simply reflected back to $A$ without otherwise disturbing it, or learning anything about its state.
\end{enumerate}

The SQKD protocol we consider in this paper was the one first presented in \cite{SQKD-first}.  A single iteration of this protocol's quantum communication stage consists of the following procedure:

\begin{enumerate}
  \item $A$ will prepare and send to $B$ a qubit of the form $\ket{0}, \ket{1}, \ket{+}$, or $\ket{-}$, choosing one at random.
  \item $B$ will choose randomly to either measure and resend or reflect the incoming qubit.
  \begin{itemize}
    \item If he chooses to measure and resend, he will save his measurement result as his raw key bit for this iteration.
  \end{itemize}
  \item $A$ will choose to measure in the same basis she originally used to prepare the qubit from step 1 (e.g., if she sent $\ket{-}$ on step 1, she will measure in the $X$ basis; if she originally sent $\ket{0}$ she will measure in the $Z$ basis).
\end{enumerate}

After repeating the above process $M$ times, $A$ will inform $B$, using the public authenticated classical channel, of her preparation basis choice from step 1 (which determines her measurement basis choice in step 3) for each iteration.  For each iteration that was performed, $B$ will inform $A$ of his choice to measure and resend, or to reflect.  If $B$ measured and resent, and if $A$ chose to prepare (step 1) and measure (step 3) in the $Z$ basis, they will use this iteration for their raw key.  In this paper, we will define $A$'s key to be her measurement result from step 3; $B$'s key will of course be his measurement result in step 2.  Another option, which we do not use in this paper, would be to define $A$'s key to be her preparation choice in step 1.

Note that, if $B$ reflects, and if $A$ prepared and measured in the $X$ basis, $A$ should expect to measure the same state she sent originally; any other result is counted as an error.  Thus, $A$ may immediately estimate the noise of the quantum channel in the $X$ basis by counting the number of iterations where $A$ measures $\ket{-}$ if $B$ reflected and she sent $\ket{+}$; similarly if she measures $\ket{+}$ but sent $\ket{-}$.  Measuring the error in the $Z$ basis may be performed easily by $B$ divulging a portion of his measurement results.  Note that this allows $A$ and $B$ to estimate the $Z$ basis error rate in both quantum channels: the forward ($A$ to $B$) and the reverse ($B$ to $A$).

Note that, to improve efficiency, we may adopt a technique used in \cite{QKD-BB84-Modification} to modify the original BB84 \cite{QKD-BB84} protocol.  Namely, we may have $A$ choose to prepare and thus measure in the $Z$ basis with greater probability than choosing the $X$ basis.  Likewise, we may have $B$ choose to measure and resend more frequently than reflecting.  The same arguments used in \cite{QKD-BB84-Modification} can apply to this setting.

\section{Security Proof}

We will now prove security against collective attacks - those where $E$ performs the same operation each iteration, but may wait to perform a measurement of her private ancilla until any future point in time of her choice.  Later we will consider security against general attacks.

Let $\mathcal{H}_T$ be the two-dimensional Hilbert space modeling the qubit (the \emph{transit space}) and let $\mathcal{H}_E$ be $E$'s private ancilla (this is finite without loss of generality) for one iteration of the protocol.  Let $(U_E, U_F)$ be a pair of unitary attack operators which both act on $\mathcal{H}_T\otimes\mathcal{H}_E$.  Here, $U_E$ will be used to attack the forward direction (the qubit traveling initially from $A$ to $B$) while $U_F$ will be used to attack the reverse direction (the qubit returning from $B$ to $A$).  This is, without loss of generality, the most general form of a collective attack $E$ may use (collective attack implies she will use the same two unitary operators $U_E$ and $U_F$ each iteration of the quantum communication stage and that $E$ will use a different ``copy'' of the space $\mathcal{H}_E$ each iteration).  In order to compute the protocol's key rate, we must first construct the density operator describing the result of a single iteration of the protocol, assuming it is used to contribute to the raw key: namely, assuming that $A$ sent a $Z$ basis state, $B$ measures and resends, and $A$ measures in the $Z$ basis.

In this event, $A$ prepares a qubit of the form $\ket{0}$ or $\ket{1}$, each chosen with probability $1/2$.  The system $E$ receives then is the mixed state:

\[
\rho_0 = \frac{1}{2}\ket{0}\bra{0}_T + \frac{1}{2}\ket{1}\bra{1}_T.
\]

We may assume, without loss of generality, that $E$'s ancilla is cleared to some ``zero'' state $\ket{0}_E \in \mathcal{H}_E$.  $E$ then attacks with $U_E$, an operator which acts on basis states as follows:

\begin{align}
U_E\ket{0,0}_{TE} &= \ket{0, e_0}_{TE} + \ket{1, e_1}_{TE}\label{eq:UE}\\
U_E\ket{1,0}_{TE} &= \ket{0, e_2}_{TE} + \ket{1, e_3}_{TE},\notag
\end{align}
where $\ket{e_i}$ are arbitrary states in $\mathcal{H}_E$, not necessarily normalized nor orthogonal.  Of course, unitarity of $U$ imposes the following conditions:
\begin{align*}
&\braket{e_0|e_0} + \braket{e_1|e_1} = 1\\
&\braket{e_2|e_2} + \braket{e_3|e_3} = 1\\
&\braket{e_0|e_2} + \braket{e_1|e_3} = 0.
\end{align*}

After this operation, $E$ passes the qubit to $B$ who performs a $Z$ basis measurement, recording his result as his raw key bit, and resending his result to $A$.  At this point, the system is in the mixed state:

\begin{align*}
\rho_1 &= \frac{1}{2}\ket{0}\bra{0}_B \otimes \left(\ket{0,e_0}\bra{0,e_0}_{TE} + \ket{0,e_2}\bra{0,e_2}_{TE}\right)\\
&+ \frac{1}{2}\ket{1}\bra{1}_B \otimes \left(\ket{1,e_1}\bra{1,e_1}_{TE} + \ket{1,e_3}\bra{1,e_3}_{TE}\right).
\end{align*}

Of course, $E$ captures the transit qubit on its return, and applies her second attack operator $U_F$.  We may write this operator's action on states of the form $\ket{i, e_j}$ as follows:
\begin{equation}\label{eq:UF}
U_F\ket{i, e_j}_{TE} = \ket{0, e_{i,j}^0}_{TE} + \ket{1, e_{i,j}^1}_{TE},
\end{equation}
where the $\ket{e_{i,j}^k}$ are arbitrary states in $\mathcal{H}_E$.  Unitarity of $U_F$ imposes several conditions on these states of course - for instance $\braket{e_j|e_j} = \braket{e_{i,j}^0|e_{i,j}^0} + \braket{e_{i,j}^1|e_{i,j}^1}$, for $i = 0,1$.

After this attack, $E$ passes the qubit to $A$ who performs a $Z$ basis measurement, using the result as her raw key bit.  Tracing out $A$'s system, leaving only $B$ and $E$'s, yields:

\begin{align}
\rho_{BE} &= \frac{1}{2}\ket{0}\bra{0}_B \otimes \left(\ket{e_{0,0}^0}\bra{e_{0,0}^0} + \ket{e_{0,0}^1}\bra{e_{0,0}^1} + \ket{e_{0,2}^0}\bra{e_{0,2}^0} + \ket{e_{0,2}^1}\bra{e_{0,2}^1}\right)\label{eq:protocol-final}\\
&+\frac{1}{2}\ket{1}\bra{1}_B \otimes \left(\ket{e_{1,1}^0}\bra{e_{1,1}^0} + \ket{e_{1,1}^1}\bra{e_{1,1}^1} + \ket{e_{1,3}^0}\bra{e_{1,3}^0} + \ket{e_{1,3}^1}\bra{e_{1,3}^1}\right).\notag
\end{align}

It is important to observe that $A$ and $B$ may estimate the $Z$ basis noise in both the forward channel and the reverse channel, during the parameter estimation stage.  In particular, they may estimate the quantity $p_{i,j,k}$ which we use to denote the probability that, if $A$ initially sends $\ket{i}$, then $B$ measures $\ket{j}$, and $A$ measures $\ket{k}$.  For example, if there is no noise in the $Z$ basis, it should hold that $p_{0,0,0} = p_{1,1,1} = 1$.  These parameters can be used to estimate the value $\braket{e_{a,b}^c|e_{a,b}^c}$.

For example, to estimate $p_{0,1,0}$, consider the case that $A$ first sends $\ket{0}$.  After $E$'s first attack, the state evolves to $\ket{0,e_0} + \ket{1,e_1}$ and the probability that $B$ measures $\ket{1}$ is $\braket{e_1|e_1}$ after which the state collapses to $\ket{1,e_1}/\sqrt{\braket{e_1|e_1}}$.  On its return, the qubit is attacked again causing it to evolve to:
\[
\frac{\ket{0,e_{1,1}^0} + \ket{1,e_{1,1}^1}}{\sqrt{\braket{e_1|e_1}}},
\]
from which it is clear that the probability of $A$ measuring $\ket{0}$ is $\braket{e_{1,1}^0|e_{1,1}^0}/\braket{e_1|e_1}$.  Combining all of this, we have (conditioning on the event that $A$ sends $\ket{0}$):
\begin{align*}
p_{0,1,0} &= Pr(B \text{ measures } \ket{1} \text{ and } A \text{ measures }\ket{0})\\
&=Pr(B \text{ measures } \ket{1}) Pr(A \text{ measures } \ket{0} | \text{ } B \text{ measures } \ket{1})\\
&=\braket{e_1|e_1} \left(\frac{\braket{e_{1,1}^0|e_{1,1}^0}}{\braket{e_1|e_1}}\right) = \braket{e_{1,1}^0|e_{1,1}^0}.
\end{align*}

Similarly, we may estimate the following:

\begin{align}
p_{0,0,0} = \braket{e_{0,0}^0|e_{0,0}^0} && p_{1,1,1} = \braket{e_{1,3}^1|e_{1,3}^1}\label{eq:prbounds}\\
p_{0,0,1} = \braket{e_{0,0}^1|e_{0,0}^1} && p_{1,1,0} = \braket{e_{1,3}^0|e_{1,3}^0}\notag\\
p_{0,1,0} = \braket{e_{1,1}^0|e_{1,1}^0} && p_{1,0,1} = \braket{e_{0,2}^1|e_{0,2}^1}\notag\\
p_{0,1,1} = \braket{e_{1,1}^1|e_{1,1}^1} && p_{1,0,0} = \braket{e_{0,2}^0|e_{0,2}^0}\notag
\end{align}

\subsection{Bounding $S(B|E)$}

Before continuing, we need a small lemma concerning the von Neumann entropy of a particular form of system.  The result is not difficult to show, however we include the proof for completeness.

\begin{lemma}\label{lemma:entropy}
Let $\mathcal{H} = \mathcal{H}_X\otimes\mathcal{H}_Y$ be a bipartite Hilbert space with $\dim\mathcal{H}_X = n$ and $\dim\mathcal{H}_Y = m$ (both finite dimensional) and let $\{\ket{1}_X, \cdots, \ket{n}_X\}$ be an orthonormal basis of $\mathcal{H}_X$.  Consider the density operator:
\[
\rho = \sum_{j=1}^np_j\ket{j}\bra{j}_X \otimes \sigma_j,
\]
acting on $\mathcal{H}$, where each $\sigma_j$ is a Hermitian, positive semi-definite operator, of unit trace, acting on $\mathcal{H}_Y$.  Furthermore, assume $tr\rho = 1$ (which implies $\sum_jp_j = 1$).  Then:
\begin{equation}\label{eq:entropy-computation}
S(\rho) = H\left(p_1, p_2, \cdots, p_n\right) + \sum_{j=1}^np_jS\left(\sigma_j\right).
\end{equation}
\end{lemma}
\begin{proof}
Choosing a suitable basis, we may write $\rho$ as:

\begin{equation}
\rho \equiv \left(\begin{array}{ccccc}
p_1\sigma_1 & 0 & 0 & \cdots & 0\\
0 & p_2\sigma_2 & 0 & \cdots & 0\\
\vdots &&&&\vdots\\
0 & 0 & 0 & \cdots & p_n\sigma_n
\end{array}\right)
\end{equation}

Let $\{\lambda_i^j\}_{i=1}^m$ be the eigenvalues of $\sigma_j$.  Since each $\sigma_j$ is Hermitian, positive semi-definite, these are real and non-negative.  Since each $\sigma_j$ is of unit trace, it holds that $\sum_i\lambda_i^j = 1$ for all $j=1, 2, \cdots, n$.  Clearly, the eigenvalues of $\rho$ then are:
\[
\bigcup_{j=1}^n\bigcup_{i=1}^m\left\{p_j\lambda_i^j\right\}.
\]

We now compute $S(\rho)$:
\begin{align*}
S(\rho) &= -\sum_{j=1}^n\sum_{i=1}^m p_j\lambda_i^j\log p_j\lambda_i^j\\
&= -\sum_{j=1}^n\sum_{i=1}^m\left(p_j\lambda_i^j\log p_j + p_j\lambda_i^j\log\lambda_i^j\right)\\
&=- \sum_{j=1}^np_j\log p_j\sum_{i=1}^m\lambda_i^j -\sum_{j=1}^np_j\sum_{i=1}^m\lambda_i^j\log\lambda_i^j\\
&=H(p_1, \cdots, p_n) + \sum_{j=1}^np_jS(\sigma_j).
\end{align*}
\end{proof}

We now return to our security analysis.  To compute the key rate, we must compute $S(B|E) = S(BE) - S(E)$ using Equation \ref{eq:protocol-final}.  Due to the high-dimensionality of $\mathcal{H}_E$, this is difficult and so we will employ a technique similar to one used in \cite{QKD-keyrate-general} for proving the security of BB84 \cite{QKD-BB84}, though suitably modified for our purposes: that is, we will condition on a new random variable of our choice, in order to simplify the analysis.  Due to the strong sub additivity of von Neumann entropy, it holds that, for any tripartite system $\mathcal{H}_X\otimes\mathcal{H}_Y\otimes\mathcal{H}_Z$:
\[
S(X|Y) \ge S(X|YZ).
\]
If we introduce a new system $\mathcal{H}_C$ into Equation \ref{eq:protocol-final}, it will hold that:
\[
S(B|E) - H(B|A) \ge S(B|EC) - H(B|A),
\]
thus providing us with a lower-bound on the key rate of this protocol.

Let $\mathcal{H}_C$ be the four dimensional space spanned by $\{\ket{C,0}, \ket{C,1}, \ket{W,1}, \ket{W,2}\}$.  We will use the state $\ket{C,i}\bra{C,i}$ to represent the event that $A$ and $B$'s raw key bits agree/match (that is, they are ``correct'') and that the qubit sent from $A$ was flipped $i$ times (in the $Z$ basis, which is all we are considering for now as $X$ basis states do not contribute to the raw key and will be considered later).  For example, if $A$ sends a $\ket{0}$, $B$ measures a $\ket{1}$, and $A$ measures a $\ket{1}$, this will be the event $\ket{C,1}\bra{C,1}$.  Similarly for the state $\ket{W,i}\bra{W,i}$ where now $A$ and $B$'s raw key bits do not match (they are ``wrong'').  (Note that in \cite{QKD-keyrate-general}, the authors only conditioned on the ``correct'' or ``wrong'' events which was sufficient in the one-way quantum channel case.)

Incorporating this new system yields the mixed state:
\begin{align}
\rho_{BEC} &= \frac{1}{2}\ket{0}\bra{0}_B\otimes&&(\ket{C,0}\bra{C,0}\otimes\ket{e_{0,0}^0}\bra{e_{0,0}^0} + \ket{W,1}\bra{W,1} \otimes \ket{e_{0,0}^1}\bra{e_{0,0}^1}\notag\\
&&&+ \ket{C,1}\bra{C,1} \otimes \ket{e_{0,2}^0}\bra{e_{0,2}^0} + \ket{W,2}\bra{W,2}\otimes\ket{e_{0,2}^1}\bra{e_{0,2}^1})\notag\\
&+ \frac{1}{2}\ket{1}\bra{1}_B\otimes&&(\ket{W,2}\bra{W,2}\otimes\ket{e_{1,1}^0}\bra{e_{1,1}^0} + \ket{C,1}\bra{C,1} \otimes \ket{e_{1,1}^1}\bra{e_{1,1}^1}\notag\\
&&&+ \ket{W,1}\bra{W,1} \otimes \ket{e_{1,3}^0}\bra{e_{1,3}^0} + \ket{C,0}\bra{C,0}\otimes\ket{e_{1,3}^1}\bra{e_{1,3}^1})\notag
\end{align}
Note that it is not relevant that $A$ and $B$ cannot know whether their key bit is ``correct'' or not (on those iterations not used for parameter estimation of course).  By conditioning on these events, however, we are able to compute a lower bound on the protocol's key rate.  The actual key rate may be higher of course.  One may think of this as providing $E$ with additional information which can only increase her power, thus providing us with a lower-bound on the protocol's security.

Choosing a suitable basis, we may write $\rho_{BEC}$ as a diagonal matrix, where the diagonal entries are elements of the form $\frac{1}{2}\braket{e_{i,j}^k|e_{i,j}^k}$ for all $\ket{e_{i,j}^k}$ which appear in the above equation.  Thus:
\begin{align}\label{eq:entropy-BEC}
S(BEC) = S(\rho_{BEC}) &= H\left(\frac{1}{2}\braket{e_{0,0}^0|e_{0,0}^0}, \cdots, \frac{1}{2}\braket{e_{1,3}^1|e_{1,3}^1}\right)\\
&=H\left(\frac{1}{2}p_{0,0,0}, \frac{1}{2}p_{0,0,1}, \cdots, \frac{1}{2}p_{1,1,1}\right),\notag
\end{align}
where the arguments in the (Shannon) entropy function above are all $p_{i,j,k}$ from Equation \ref{eq:prbounds}.  This is a quantity that $A$ and $B$ may compute after parameter estimation.

What remains is to bound $S(EC)$.  Tracing out $B$ from $\rho_{BEC}$ yields:
\begin{align*}
\rho_{EC} &= \ket{C,0}\bra{C,0} \otimes \left(\frac{1}{2}\sigma_1\right) + \ket{C,1}\bra{C,1}\otimes \left(\frac{1}{2}\sigma_2\right)\\
&+ \ket{W,1}\bra{W,1} \otimes \left(\frac{1}{2}\sigma_3\right) + \ket{W,2}\bra{W,2}\otimes \left(\frac{1}{2}\sigma_4\right),
\end{align*}
where $\sigma_i$ are the positive semi-definite operators:
\begin{align*}
\sigma_1 &= \ket{e_{0,0}^0}\bra{e_{0,0}^0} + \ket{e_{1,3}^1}\bra{e_{1,3}^1}\\
\sigma_2 &= \ket{e_{0,2}^0}\bra{e_{0,2}^0} + \ket{e_{1,1}^1}\bra{e_{1,1}^1}\\
\sigma_3 &= \ket{e_{0,0}^1}\bra{e_{0,0}^1} + \ket{e_{1,3}^0}\bra{e_{1,3}^0}\\
\sigma_4 &= \ket{e_{0,2}^1}\bra{e_{0,2}^1} + \ket{e_{1,1}^0}\bra{e_{1,1}^0}
\end{align*}

Assume, for now, that $tr\sigma_j > 0$ for all $j=1, 2,3,4$.  Let $t_j = tr\sigma_j$ and define $\tilde{\sigma}_j = \sigma_j/t_j$.  Then we may write $\rho_{EC}$ as:
\begin{align*}
\rho_{EC} &= \ket{C,0}\bra{C,0} \otimes \left(\frac{1}{2}t_1\tilde{\sigma}_1\right) + \ket{C,1}\bra{C,1}\otimes \left(\frac{1}{2}t_2\tilde{\sigma}_2\right)\\
&+ \ket{W,1}\bra{W,1} \otimes \left(\frac{1}{2}t_3\tilde{\sigma}_3\right) + \ket{W,2}\bra{W,2}\otimes \left(\frac{1}{2}t_4\tilde{\sigma}_4\right).
\end{align*}

Each $\tilde{\sigma_j}$ is a positive semi-definite operator of unit trace.  Also, since $tr\rho_{EC} = 1$ implies $\frac{1}{2}t_1+\cdots+\frac{1}{2}t_4 = 1$, we may apply Lemma \ref{lemma:entropy} to the above state yielding:
\[
S(EC) = S(\rho_{EC}) = H\left(\frac{1}{2}t_1, \cdots, \frac{1}{2}t_4\right) + \frac{1}{2}\sum_{j=1}^4t_jS(\tilde{\sigma}_j).
\]

We assumed $t_j > 0$ above.  If there is a $t_j = 0$, then observe $\sigma_j \equiv 0$.  Indeed:
\begin{align*}
tr\sigma_j = 0 \iff \braket{e_{x,y}^z|e_{x,y}^z} + \braket{e_{a,b}^c|e_{a,b}^c} = 0,
\end{align*}
for appropriate $a,b,c, x,y,z$.  But, since $\braket{e|e} \ge 0$ for any vector $e$, this forces $\braket{e_{x,y}^z|e_{x,y}^z} = \braket{e_{a,b}^c|e_{a,b}^c} = 0$ which is true only if $\ket{e_{x,y}^z} \equiv \ket{e_{a,b}^c} \equiv 0$.  Thus, if $tr\sigma_j = 0$, then $\sigma_j \equiv 0$, and so it could simply be removed from the description of $\rho_{EC}$ above, and any reference to the $j$'th matrix is removed from the subsequent computation of $S(EC)$.

Of course $t_1 = tr\sigma_1 = \braket{e_{0,0}^0|e_{0,0}^0} + \braket{e_{1,3}^1|e_{1,3}^1} = p_{0,0,0} + p_{1,1,1}$, and similarly for the other $\sigma_j$.  Observing that each $\sigma_j$ is a two-dimensional system, we may use the trivial bound $S(\tilde{\sigma}_j) \le 1$ to show:
\begin{align*}
S(EC) &\le H\left(\frac{1}{2}(p_{0,0,0} + p_{1,1,1}), \frac{1}{2}(p_{1,0,0} + p_{0,1,1}), \frac{1}{2}(p_{0,0,1} + p_{1,1,0}), \frac{1}{2}(p_{1,0,1} + p_{0,1,0})\right)\\
&+ \frac{1}{2}\left(p_{1,0,0}+p_{0,1,1} + p_{0,0,1} + p_{1,1,0} + p_{1,0,1} + p_{0,1,0}\right)\\
&+\frac{1}{2}(p_{0,0,0}+p_{1,1,1})S(\tilde{\sigma}_1).
\end{align*}
Note that this bound holds even if there is a $t_j = 0$.  Thus, at this point, there is no need to take extra care of such a case.

If the noise of the quantum channel is low, the values $p_{i,j,k}$ should be low, except for $p_{0,0,0}$ and $p_{1,1,1}$ which should be high.  All that remains, therefore, is to upper bound $S(\tilde{\sigma}_1)$.

Let us first find the eigenvalues of $\sigma_1$ (the unnormalized version); the eigenvalues of $\tilde{\sigma}_1$ then will simply be scalar multiples of these.  We may write $\ket{e_{0,0}^0} = \sqrt{p_{0,0,0}}\ket{e}$ and $\ket{e_{1,3}^1} = \alpha\ket{e} + \beta\ket{\zeta}$, where $\braket{e|e} = \braket{\zeta|\zeta} = 1$, $\braket{e|\zeta} = 0$, and $\alpha, \beta \in \mathbb{C}$.  This further implies:
\begin{align}
&|\alpha|^2 + |\beta|^2 = \braket{e_{1,3}^1|e_{1,3}^1} = p_{1,1,1}\\
&\alpha\sqrt{p_{0,0,0}} = \braket{e_{0,0}^0|e_{1,3}^1} \Rightarrow |\alpha|^2 = \frac{|\braket{e_{0,0}^0|e_{1,3}^1}|^2}{p_{0,0,0}}\label{eq:alpha2}
\end{align}
(we may assume that $p_{0,0,0} > 0$; otherwise, there is too much noise, and $A$ and $B$ will abort).  In this $\{\ket{e}, \ket{\zeta}\}$ basis, we may write $\sigma_1$ as:
\[
\sigma_1 \equiv \left(\begin{array}{cc}p_{0,0,0}+|\alpha|^2 & \alpha\beta^*\\\\
\alpha^*\beta & |\beta|^2\end{array}\right).
\]

The eigenvalues of $\sigma_1$, denoted $\lambda_+$ and $\lambda_-$, are:
\begin{align*}
\lambda_{\pm} &= \frac{1}{2}\left(p_{0,0,0} + p_{1,1,1} \pm \sqrt{\left(p_{0,0,0} + |\alpha|^2 - |\beta|^2\right)^2 + 4|\alpha|^2|\beta|^2}\right)\\
&=\frac{1}{2}\left(p_{0,0,0}+p_{1,1,1} \pm \sqrt{\left(p_{0,0,0} - p_{1,1,1} + 2|\alpha|^2\right)^2 + 4|\alpha|^2|\beta|^2}\right),
\end{align*}
where, above, we have used the identity $|\alpha|^2 + |\beta|^2 = p_{1,1,1} \Rightarrow -|\beta|^2 = |\alpha|^2 - p_{1,1,1}$.  Let $\Delta = p_{0,0,0} - p_{1,1,1}$ and, using the identity $|\beta|^2 = p_{1,1,1} - |\alpha|^2$, we continue:
\begin{align*}
\lambda_{\pm} &= \frac{1}{2}\left(p_{0,0,0} + p_{1,1,1} \pm \sqrt{\Delta^2 + 4|\alpha|^4 + 4\Delta|\alpha|^2 + 4|\alpha|^2(p_{1,1,1}-|\alpha|^2)}\right)\\
&=\frac{1}{2}\left(p_{0,0,0} + p_{1,1,1} \pm \sqrt{\Delta^2 + 4|\alpha|^2(|\alpha|^2 + \Delta + p_{1,1,1} - |\alpha|^2)}\right)\\
&=\frac{1}{2}\left(p_{0,0,0} + p_{1,1,1} \pm \sqrt{(p_{0,0,0}-p_{1,1,1})^2 + 4|\alpha|^2p_{0,0,0}}\right).
\end{align*}
Finally, using Equation \ref{eq:alpha2}, we have:
\begin{equation}
\lambda_{\pm} = \frac{1}{2}\left(p_{0,0,0} + p_{1,1,1} \pm \sqrt{(p_{0,0,0} - p_{1,1,1})^2 + 4|\braket{e_{0,0}^0|e_{1,3}^1}|^2}\right).
\end{equation}

From this, the eigenvalues of $\tilde{\sigma}_1 = \sigma_1 / tr\sigma_1 = \sigma_1 / (p_{0,0,0}+p_{1,1,1})$ are:
\begin{equation}\label{eq:eigenvalue-sigma1}
\tilde{\lambda}_{\pm} = \frac{1}{2} \pm \frac{\sqrt{(p_{0,0,0} - p_{1,1,1})^2 + 4|\braket{e_{0,0}^0|e_{1,3}^1}|^2}}{2(p_{0,0,0} + p_{1,1,1})},
\end{equation}
and so:
\[
S(\tilde{\sigma}_1) = -\tilde{\lambda}_+\log\tilde{\lambda}_+ - \tilde{\lambda}_-\log\tilde{\lambda}_- = h(\tilde{\lambda}_+),
\]
a function which depends only on the quantity $|\braket{e_{0,0}^0|e_{1,3}^1}|^2 \ge 0$.  Note that, as $|\braket{e_{0,0}^0|e_{1,3}^1}|^2$ decreases to zero, this causes $\tilde{\lambda}_\pm$ to become closer to $1/2$, which causes $S(\tilde{\sigma}_1)$ to increase (this function taking its maximum when $\tilde{\lambda}_{\pm} = \frac{1}{2}$).  Thus, to find a lower bound for the key rate $r \ge S(BEC) - S(EC) - H(B|A)$, we must find a lower bound on $|\braket{e_{0,0}^0|e_{1,3}^1}|^2$ (thus upper bounding $S(EC)$).  Indeed, if $|\braket{e_{0,0}^0|e_{1,3}^1}|^2 \ge \mathcal{B}$, then $S(\tilde{\sigma}_1) = h(\tilde{\lambda}_+) \le h(\tilde{\lambda})$, where:
\begin{equation}\label{eq:ev-max-bound}
\tilde{\lambda} = \frac{1}{2} + \frac{\sqrt{(p_{0,0,0} - p_{1,1,1})^2 + 4\mathcal{B}}}{2(p_{0,0,0} + p_{1,1,1})}.
\end{equation}

\subsection{Using the $X$ Basis Noise}

We will lower bound the value $|\braket{e_{0,0}^0|e_{1,3}^1}|^2$, by considering the noise in the $X$ basis (note that, thus far, we have considered only the noise in the $Z$ basis).  Assume now that $A$ sends an $X$ basis state $\ket{+}$ or $\ket{-}$ initially, $B$ chooses to reflect, and $A$ measures in the $X$ basis, thus allowing her to estimate the channel noise in this basis.

In this event, $B$'s operation is essentially the identity operator and, so, if $A$ sends the state $\ket{a}$ (either $\ket{+}$ or $\ket{-}$ in our case), the state returning to her, after it passed through $E$ twice, is simply $\ket{a'} = V\ket{a}$, where $V = U_FU_E$ (the same operators used last section).  Using Equations \ref{eq:UE} and \ref{eq:UF}, we may describe $V$'s action on basis states $\ket{0}, \ket{1} \in \mathcal{H}_T$ as follows (as before, assuming, without loss of generality, that $E$'s ancilla is cleared to the zero state $\ket{0}_E$):
\begin{align*}
V\ket{0,0}_{TE} &= U_F(\ket{0,e_0} + \ket{1,e_1})\\
&= \ket{0}\otimes(\underbrace{\ket{e_{0,0}^0} + \ket{e_{1,1}^0}}_{\ket{f_0}}) + \ket{1}\otimes(\underbrace{\ket{e_{0,0}^1} + \ket{e_{1,1}^1}}_{\ket{f_1}})\\
&=\ket{0,f_0} + \ket{1,f_1}\\\\
V\ket{1,0}_{TE} &= U_F(\ket{0,e_2} + \ket{1,e_3})\\
&=\ket{0}\otimes(\underbrace{\ket{e_{0,2}^0} + \ket{e_{1,3}^0}}_{\ket{f_2}}) + \ket{1}\otimes(\underbrace{\ket{e_{0,2}^1} + \ket{e_{1,3}^1}}_{\ket{f_3}})\\
&=\ket{0,f_2} + \ket{1,f_3}.
\end{align*}

Since $U_E$ and $U_F$ are both unitary, so is $V = U_FU_E$ which implies:
\begin{align}
&\braket{f_0|f_0} + \braket{f_1|f_1} = 1\label{eq:UEUFunitary}\\
&\braket{f_2|f_2} + \braket{f_3|f_3} = 1\notag\\
&\braket{f_0|f_2} + \braket{f_1|f_3} = 0.\notag
\end{align}

By linearity, we have:
\begin{align}
V\ket{+,0}_{TE} &=\ket{+, g_0} + \ket{-, g_1}\label{eq:UEUFX}\\
V\ket{-,0}_{TE} &= \ket{+, g_2} + \ket{-, g_3},\notag
\end{align}
where:
\begin{align}
\ket{g_0} &= \frac{1}{2}(\ket{f_0} + \ket{f_1} + \ket{f_2} + \ket{f_3})\label{eq:UEUFXg}\\
\ket{g_1} &= \frac{1}{2}(\ket{f_0} - \ket{f_1} + \ket{f_2} - \ket{f_3})\notag\\
\ket{g_2} &= \frac{1}{2}(\ket{f_0} + \ket{f_1} - \ket{f_2} - \ket{f_3})\notag\\
\ket{g_3} &= \frac{1}{2}(\ket{f_0} - \ket{f_1} - \ket{f_2} + \ket{f_3}).\notag
\end{align}

In this notation, the probability that $A$ measures $\ket{-}$ if she originally sent $\ket{+}$ and $B$ reflected, is $\braket{g_1|g_1}$; similarly, $\braket{g_2|g_2}$ is the probability that $A$ measures $\ket{+}$ if she originally sent $\ket{-}$.  These quantities, which $A$ may estimate in the parameter estimation stage, represent the error $E$'s attack induces in the $X$ basis.

Let $p_{+-}$ be the probability that $A$ measures $\ket{-}$ if she sends $\ket{+}$ (assuming $B$ reflected); similarly define $p_{-+}$ as the probability that $A$ measures $\ket{+}$ if she initially sent $\ket{-}$.  Then, using Equations \ref{eq:UEUFunitary}, \ref{eq:UEUFX}, and \ref{eq:UEUFXg}, we have:
\begin{align*}
p_{+-} = \braket{g_1|g_1} &= \frac{1}{2} + \frac{1}{2}Re(-\braket{f_0|f_1} - \braket{f_0|f_3} - \braket{f_1|f_2} - \braket{f_2|f_3})\\
p_{-+} = \braket{g_2|g_2} &= \frac{1}{2} + \frac{1}{2}Re(\braket{f_0|f_1} - \braket{f_0|f_3} - \braket{f_1|f_2} + \braket{f_2|f_3}).
\end{align*}

Summing these two and expanding yields:
\begin{align*}
p_{+-} + p_{-+} &= 1 - Re(\braket{f_0|f_3} + \braket{f_1|f_2})\\
&= 1 - Re(\braket{e_{0,0}^0|e_{0,2}^1} + \braket{e_{0,0}^0|e_{1,3}^1} + \braket{e_{1,1}^0|e_{0,2}^1} + \braket{e_{1,1}^0|e_{1,3}^1})\\
&-Re(\braket{e_{0,0}^1|e_{0,2}^0} + \braket{e_{0,0}^1|e_{1,3}^0} + \braket{e_{1,1}^1|e_{0,2}^0} + \braket{e_{1,1}^1|e_{1,3}^0}).
\end{align*}
Solving for $Re(\braket{e_{0,0}^0|e_{1,3}^1}$ provides us with the expression:
\begin{align*}
Re(\braket{e_{0,0}^0|e_{1,3}^1}) &= 1 - p_{+-} - p_{-+} - Re(\braket{e_{0,0}^0|e_{0,2}^1} + \braket{e_{1,1}^0|e_{0,2}^1})\\
&-Re(\braket{e_{1,1}^0|e_{1,3}^1} + \braket{e_{0,0}^1|e_{0,2}^0} + \braket{e_{0,0}^1|e_{1,3}^0})\\
&-Re(\braket{e_{1,1}^1|e_{0,2}^0} + \braket{e_{1,1}^1|e_{1,3}^0})
\end{align*}

Observe that, for any two vectors $\ket{x}$ and $\ket{y}$, it holds: $|Re(\braket{x|y})| \le |\braket{x|y}| \le \sqrt{\braket{x|x}\braket{y|y}}$, the first inequality is obvious, the last inequality is due to the Cauchy-Schwarz inequality.  Thus, $Re(\braket{x|y}) \in [-\sqrt{\braket{x|x}\braket{y|y}}, \sqrt{\braket{x|x}\braket{y|y}}]$.  Using this fact, and Equation \ref{eq:prbounds}, yields:
\begin{align}
Re(\braket{e_{0,0}^0|e_{1,3}^1}) \ge 1 &- p_{+-} - p_{-+} - \sqrt{p_{0,0,0}p_{1,0,1}} - \sqrt{p_{0,1,0}p_{1,0,1}}\label{eq:XBound}\\
&- \sqrt{p_{0,1,0}p_{1,1,1}} - \sqrt{p_{0,0,1}p_{1,0,0}} - \sqrt{p_{0,0,1}p_{1,1,0}}\notag\\
&- \sqrt{p_{0,1,1}p_{1,0,0}} - \sqrt{p_{0,1,1}p_{1,1,0}} = B.\notag
\end{align}

Observing that $|\braket{e_{0,0}^0|e_{1,3}^1}|^2 = Re^2(\braket{e_{0,0}^0|e_{1,3}^1}) + Im^2(\braket{e_{0,0}^0|e_{1,3}^1}) \ge Re^2(\braket{e_{0,0}^0|e_{1,3}^1})$, the right hand side of Equation \ref{eq:XBound} (denoted $B$), assuming it is non-negative (which should be the case if the noise is small enough), may be used to lower bound $|\braket{e_{0,0}^0|e_{1,3}^1}|^2$, thus upper bounding $S(EC)$ as required.

Indeed, let:
\[
\mathcal{B} = \left\{\begin{array}{ll}B^2 & \text{if } B \ge 0\\ 0 & \text{otherwise}\end{array}\right.,
\]
then, from the above discussion, it is clear that $|\braket{e_{0,0}^0|e_{1,3}^1}|^2 \ge \mathcal{B}$ (note that $|\braket{e_{0,0}^0|e_{1,3}}|^2$ is always non-negative, so it makes sense to ``cap'' $\mathcal{B}$ at zero in the event $B < 0$).  Using this, with Equation \ref{eq:ev-max-bound} and the discussion immediately above it, gives us a bound on $S(\tilde{\sigma}_1)$ and thus a bound on the quantity $S(EC)$.

\subsection{Final Key Rate Bound}
All that remains is to compute $H(B|A) = H(B,A) - H(A)$.  However, this is easily done given $A$ and $B$'s estimate of the $Z$ basis error rate.  Indeed, let $p_A(0)$ be the probability that $A$'s raw key bit is zero.  This value is simply:
\begin{equation}\label{eq:pa}
p_A(0) = \frac{1}{2}(p_{0,0,0} + p_{0,1,0} + p_{1,1,0} + p_{1,0,0}),
\end{equation}
and so:
\[
H(A) = H(p_A(0), 1-p_A(0)) = h(p_A(0)).
\]

Let $p(b,a)$ be the probability that $B$'s raw key bit is $b$ while $A$'s is $a$.  These values are:
\begin{align}
p(0,0) = \frac{1}{2}(p_{0,0,0} + p_{1,0,0}) && p(1,1) = \frac{1}{2}(p_{0,1,1} + p_{1,1,1})\label{eq:jointPr}\\
p(0,1) = \frac{1}{2}(p_{0,0,1} + p_{1,0,1}) && p(1,0) = \frac{1}{2}(p_{0,1,0} + p_{1,1,0}),\notag
\end{align}
from which the value of $H(B,A)$ may be computed directly.

Putting everything together, the key rate $r$ is lower-bounded by:
\begin{align}
r \ge &H\left(\frac{1}{2}p_{0,0,0}, \frac{1}{2}p_{0,0,1}, \cdots, \frac{1}{2}p_{1,1,1}\right)\label{eq:keyrate}\\
&-H\left(\frac{1}{2}(p_{0,0,0} + p_{1,1,1}), \frac{1}{2}(p_{1,0,0} + p_{0,1,1}), \frac{1}{2}(p_{0,0,1} + p_{1,1,0}), \frac{1}{2}(p_{1,0,1} + p_{0,1,0})\right)\notag\\
&- \frac{1}{2}\left(p_{1,0,0}+p_{0,1,1} + p_{0,0,1} + p_{1,1,0} + p_{1,0,1} + p_{0,1,0}\right)\notag\\
&-\frac{1}{2}(p_{0,0,0}+p_{1,1,1})h(\tilde{\lambda})\notag\\
&+h\left(p_A(0)\right) - H\left(p(0,0), p(0,1), p(1,0), p(1,1)\right),\notag
\end{align}
where $\tilde{\lambda}$ is from Equation \ref{eq:ev-max-bound}.

This key rate equation is a function, easily computed, depending only on parameters that may be estimated by $A$ and $B$.

\subsection{Security Against General Attacks}

The above proves security against collective attacks.  However, after the protocol, $A$ and $B$ may symmetrize their raw key by permuting it using a randomly chosen, and publicly disclosed, permutation.  This makes the protocol permutation invariant, in which case, as shown in \cite{QKD-general-attack,QKD-general-attack2}, security against collective attacks is sufficient to prove security against any arbitrary general attack.  Thus, we have proven this protocol's unconditional security.

\subsection{Examples}

Our work above allows $A$ and $B$ to compute the final fraction of secure secret key bits that they can distill after privacy amplification, using only the observed statistics $p_{i,j,k}$, $p_{+-}$, and $p_{-+}$.  Let us now demonstrate our key rate bound on certain examples.  In particular, let us assume that $E$'s attack is symmetric in that it can be characterized as follows:
\begin{enumerate}
  \item Let $\overrightarrow{Q}$ denote the probability that if $A$ sends $\ket{i}$ initially, $B$ measures $\ket{1-i}$, for $i=0,1$.
  \item Let $\overleftarrow{Q}$ denote the probability that if $B$ sends $\ket{i}$, then $A$ measures $\ket{1-i}$, for $i=0,1$ (independently of the first channel).
  \item Let $p_{+-} = p_{-+} = Q_X$.
\end{enumerate}

We remark that $A$ and $B$ may estimate these three parameters and, in fact, can even enforce the restriction that $E$ use such a symmetric attack, a strategy used in other fully quantum protocols \cite{QKD-keyrate-general} (though in those protocols, there was only one quantum channel to consider).

In this case we have:
\begin{align*}
p_{0,0,0} &= p_{1,1,1} = (1-\overrightarrow{Q})(1-\overleftarrow{Q})\\
p_{0,0,1} &= p_{1,1,0} = (1-\overrightarrow{Q})\overleftarrow{Q}\\
p_{0,1,0} &= p_{1,0,1} = \overrightarrow{Q}\overleftarrow{Q}\\
p_{0,1,1} &= p_{1,0,0} = \overrightarrow{Q}(1-\overleftarrow{Q})
\end{align*}

Recall that, so long as $r > 0$, $A$ and $B$ may distill a secure secret key.  We consider three scenarios (summarized in Table \ref{table:noise}):

\begin{table}
\begin{tabular}{l|c|c|c}
&$\overrightarrow{Q} = \overleftarrow{Q} = Q$ & $\overrightarrow{Q} = Q/2$, $\overleftarrow{Q} = Q$ & $\overrightarrow{Q} = Q$ and $\overleftarrow{Q} = Q/2$\\
\hline
$Q_X = Q/2$ & $Q \le 5.92\%$ & $Q \le 6.98\%$ & $Q \le 8.96\%$\\
$Q_X = Q$ & $Q \le 5.34\%$ & $Q \le 6.16\%$ & $Q \le 7.79\%$\\
$Q_X = 2Q$ & $Q \le 4.51\%$ & $Q \le 5.05\%$ & $Q \le 6.25\%$
\end{tabular}
\caption{Showing the maximal value of $Q$, in a variety of scenarios, for which the key rate (Equation \ref{eq:keyrate}) remains positive.  $\protect\overrightarrow{Q}$ represents the probability that a $\ket{i}$ flips to a $\ket{1-i}$ in the forward direction ($A$ to $B$) while $\protect\overleftarrow{Q}$ represents the probability that a $\ket{i}$ flips to $\ket{1-i}$ in the reverse direction ($B$ to $A$).  Finally, $Q_X$ is the probability that $A$ measures $\ket{\pm}$ if she sent $\ket{\mp}$ and $B$ reflects.}\label{table:noise}
\end{table}

\begin{enumerate}
  \item $\overrightarrow{Q} = \overleftarrow{Q} = Q$.  In this case:
  \begin{itemize}
    \item If $Q_X = Q$, we see that $r > 0$ for all $Q \le 5.34\%$.
    \item If $Q_X = 2Q$, then $r > 0$ for all $Q \le 4.51\%$.
    \item If $Q_X = Q/2$, then the key rate, $r$, is positive for all $Q \le 5.92\%$.
  \end{itemize}
  See Figure \ref{fig:equal}.
  
  \item $\overrightarrow{Q} = Q/2$, $\overleftarrow{Q} = Q$.  In this case:
  \begin{itemize}
    \item If $Q_X = Q$, then $r > 0$ for all $Q \le 6.16\%$.
    \item If $Q_X = 2Q$, the key rate is positive for $Q \le 5.05\%$.
    \item If $Q_X = Q/2$, then $r > 0$ for $Q \le 6.98\%$.
  \end{itemize}
  See Figure \ref{fig:norighterror}.
  
  \item $\overrightarrow{Q} = Q$ and $\overleftarrow{Q} = Q/2$.  Then:
  \begin{itemize}
    \item If $Q_X = Q$, we see $r$ is positive for $Q \le 7.79\%$.
    \item If $Q_X = 2Q$ then $r > 0$ for $Q \le 6.25\%$.
    \item If $Q_X = Q/2$ then $r > 0$ for $Q \le 8.96\%$.
  \end{itemize}
  See Figure \ref{fig:nolefterror}.
\end{enumerate}

\begin{figure}
\includegraphics[width=300pt]{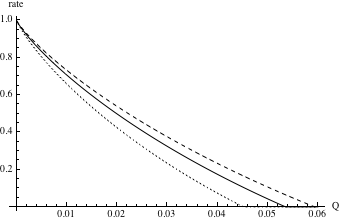}
\caption{Key rate bound when the $Z$ basis noise in the forward channel is the same as the noise in the reverse; that is, $\protect\overrightarrow{Q} = \protect\overleftarrow{Q} = Q$.  Solid line is when $Q_X = Q$ (where $Q_X$ is the $X$ basis noise).  Dashed line is when $Q_X = Q/2$; finally, the dotted line (lower) is for $Q_X = 2Q$.}\label{fig:equal}
\end{figure}

\begin{figure}
\includegraphics[width=300pt]{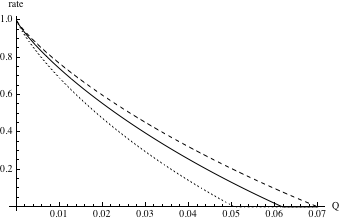}
\caption{Key rate bound when $\protect\overrightarrow{Q} = Q/2$ while $\protect\overleftarrow{Q} = Q$.  Solid line is when $Q_X = Q$.  Dashed line is when $Q_X = Q/2$; finally, the dotted line (lower) is for $Q_X = 2Q$.}\label{fig:norighterror}
\end{figure}

\begin{figure}
\includegraphics[width=300pt]{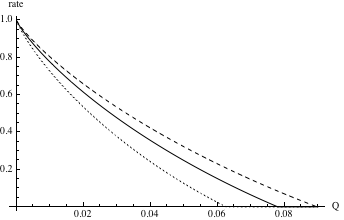}
\caption{Key rate bound when $\protect\overrightarrow{Q} = Q$ while $\protect\overleftarrow{Q} = Q/2$.  Solid line is when $Q_X = Q$.  Dashed line is when $Q_X = Q/2$; finally, the dotted line (lower) is for $Q_X = 2Q$.}\label{fig:nolefterror}
\end{figure}

Our results clearly show that one of the important factors (though not the only one) to this SQKD protocol's key rate, is the noise in the return quantum channel, connecting $B$ to $A$.  This makes sense, as any noise in the forward channel does not directly lead to an error in $A$ and $B$'s raw key bit and, thus, does not lead to additional information leaking due to error correction (though an attack in the forward direction might increase $E$'s information - a factor our key rate bound takes into account; in fact, a particular two-way eavesdropping strategy was shown in \cite{SQKD-twoway-eavesdropping} which provided a greater advantage to $E$ than by her simply attacking a single channel).

Interestingly, considering the case when $\overrightarrow{Q} = \overleftarrow{Q} = Q_X = Q$, even though this protocol cannot withstand as high an error rate as BB84 \cite{QKD-BB84}, which maintains a positive key rate for $Q = Q_X \le 11\%$; this protocol is comparable to a three state variant of BB84 \cite{QKD-BB84-three-state} which can withstand up to $4.25\%$ in this scenario.  It is also comparable to B92 \cite{QKD-B92} which can withstand up to $4.8\%$ error (assuming no preprocessing and a depolarization channel) before the key rate drops to zero \cite{QKD-renner-keyrate}.  Thus, this proves that even though we are limiting $B$'s capabilities in the semi-quantum setting, we are still capable of getting comparable tolerable noise thresholds compared to ``fully'' quantum protocols.

\section{Closing Remarks}

We have proven, for the first time, the unconditional security of a multi-state semi-quantum key distribution protocol.  Our bound may not be tight due to our conditioning on the additional random variable $C$.  However, it does show that this protocol is comparable to B92 \cite{QKD-B92} and the three-state BB84 \cite{QKD-BB84-three-state}, at least in terms of maximally allowed noise in the perfect qubit scenario.

We also showed that this protocol is most sensitive to noise in the reverse quantum channel (connecting $B$ to $A$).  Thus, if a practical implementation were constructed, one may make less effort to control the noise in the forward direction channel than the reverse.

It might be possible to adapt the technique we used in this paper to the security proof of other SQKD protocols (for instance the three and two state protocols described in \cite{SQKD-lessthan4}).  It would be interesting to compare the key rate bounds of these various protocols.


\begin{thebibliography}{10}

\bibitem{QKD-survey}
Valerio Scarani, Helle Bechmann-Pasquinucci, Nicolas~J. Cerf, Miloslav
  Du\ifmmode~\check{s}\else \v{s}\fi{}ek, Norbert L\"utkenhaus, and Momtchil
  Peev.
\newblock The security of practical quantum key distribution.
\newblock {\em Rev. Mod. Phys.}, 81:1301--1350, Sep 2009.

\bibitem{QKD-renner-keyrate}
Renato Renner, Nicolas Gisin, and Barbara Kraus.
\newblock Information-theoretic security proof for quantum-key-distribution
  protocols.
\newblock {\em Phys. Rev. A}, 72:012332, Jul 2005.

\bibitem{QKD-Winter-keyrate}
Igor Devetak and Andreas Winter.
\newblock Distillation of secret key and entanglement from quantum states.
\newblock {\em Proceedings of the Royal Society A: Mathematical, Physical and
  Engineering Science}, 461(2053):207--235, 2005.

\bibitem{QKD-general-attack}
Matthias Christandl, Robert Konig, and Renato Renner.
\newblock Postselection technique for quantum channels with applications to
  quantum cryptography.
\newblock {\em Phys. Rev. Lett.}, 102:020504, Jan 2009.

\bibitem{QKD-general-attack2}
Renato Renner.
\newblock Symmetry of large physical systems implies independence of
  subsystems.
\newblock {\em Nature Physics}, 3(9):645--649, 2007.

\bibitem{SQKD-first}
Michel Boyer, D.~Kenigsberg, and T.~Mor.
\newblock Quantum key distribution with classical bob.
\newblock In {\em Quantum, Nano, and Micro Technologies, 2007. ICQNM '07. First
  International Conference on}, pages 10--10, 2007.

\bibitem{SQKD-second}
Michel Boyer, Ran Gelles, Dan Kenigsberg, and Tal Mor.
\newblock Semiquantum key distribution.
\newblock {\em Phys. Rev. A}, 79:032341, Mar 2009.

\bibitem{SQKD-3}
Wang Jian, Zhang Sheng, Zhang Quan, and Tang Chao-Jing.
\newblock Semiquantum key distribution using entangled states.
\newblock {\em Chinese Physics Letters}, 28(10):100301, 2011.

\bibitem{SQKD-lessthan4}
Xiangfu Zou, Daowen Qiu, Lvzhou Li, Lihua Wu, and Lvjun Li.
\newblock Semiquantum-key distribution using less than four quantum states.
\newblock {\em Phys. Rev. A}, 79:052312, May 2009.

\bibitem{SQKD-cl-A}
Hua Lu and Qing-Yu Cai.
\newblock Quantum key distribution with classical alice.
\newblock {\em International Journal of Quantum Information}, 6(06):1195--1202,
  2008.

\bibitem{SQKD-information}
Takayuki Miyadera.
\newblock Relation between information and disturbance in quantum key
  distribution protocol with classical alice.
\newblock {\em Int. J. of Quantum Information}, 9, 2011.

\bibitem{SQKD-MultiUser}
W.~O. Krawec.
\newblock Mediated semi-quantum key distribution.
\newblock {\em arXiv preprint arXiv:1411.6024}, 2014.

\bibitem{SQKD-Single-Security}
W.O. Krawec.
\newblock Restricted attacks on semi-quantum key distribution protocols.
\newblock {\em Quantum Information Processing}, 13(11):2417--2436, 2014.

\bibitem{QKD-B92}
Charles~H. Bennett.
\newblock Quantum cryptography using any two nonorthogonal states.
\newblock {\em Phys. Rev. Lett.}, 68:3121--3124, May 1992.

\bibitem{QKD-BB84-three-state}
Chi-Hang~Fred Fung and Hoi-Kwong Lo.
\newblock Security proof of a three-state quantum-key-distribution protocol
  without rotational symmetry.
\newblock {\em Phys. Rev. A}, 74:042342, Oct 2006.

\bibitem{QKD-BB84-Modification}
Hoi-Kwong Lo, Hoi-Fung Chau, and M~Ardehali.
\newblock Efficient quantum key distribution scheme and a proof of its
  unconditional security.
\newblock {\em Journal of Cryptology}, 18(2):133--165, 2005.

\bibitem{QKD-BB84}
Charles~H Bennett and Gilles Brassard.
\newblock Quantum cryptography: Public key distribution and coin tossing.
\newblock In {\em Proceedings of IEEE International Conference on Computers,
  Systems and Signal Processing}, volume 175. New York, 1984.

\bibitem{QKD-keyrate-general}
Matthias Christandl, Renato Renner, and Artur Ekert.
\newblock A generic security proof for quantum key distribution.
\newblock {\em arXiv preprint quant-ph/0402131}, 2004.

\bibitem{SQKD-twoway-eavesdropping}
Arpita Maitra and Goutam Paul.
\newblock Eavesdropping in semiquantum key distribution protocol.
\newblock {\em Information Processing Letters}, 113(12):418--422, 2013.

\end{thebibliography}

\end{document}